\documentclass[11pt]{article}
\usepackage{amssymb}
\usepackage[perpage,symbol]{footmisc}
\usepackage{listings}
\usepackage{amsfonts}
\usepackage{enumerate}
\usepackage{amsmath}
\usepackage{cite}
\usepackage{array}
\usepackage{booktabs}
\usepackage{longtable}
\usepackage{rotating}
\usepackage{multirow}

\renewcommand{\thefootnote}{\arabic{footnote}}
\author{\renewcommand{\thefootnote}{\arabic{footnote}} Yan Liu\footnotemark$\phantom{\;}^,$\footnotemark, Xiwang Cao\footnotemark}
\begin{document}

\newtheorem{theorem}{Theorem}[section]
\newtheorem{corollary}[theorem]{Corollary}
\newtheorem{definition}[theorem]{Definition}
\newtheorem{proposition}[theorem]{Proposition}
\newtheorem{lemma}[theorem]{Lemma}
\newtheorem{example}[theorem]{Example}
\newenvironment{proof}{\noindent {\bf Proof.}}{\rule{3mm}{3mm}\par\medskip}
\newcommand{\remark}{\medskip\par\noindent {\bf Remark.~~}}

\title{Optimal $p$-ary cyclic codes with two zeros}
\renewcommand{\thefootnote}{\arabic{footnote}}
\footnotetext[1]{College of Mathematics and Physics,  Yancheng Institute of Technology, Yancheng, 224003, China,  liuyan0916@126.com. Y. Liu  is supported by the Foundation of Yancheng Institute of Technology (No. XJ201746).}
\footnotetext[2]{ School of Mathematical Sciences, Nanjing University of Aeronautics and Astronautics, Nanjing, 210016, China.}
\renewcommand{\thefootnote}{\arabic{footnote}} \footnotetext[3]{Corresponding author, School of Mathematical Sciences, Nanjing University of Aeronautics and Astronautics, Nanjing, 210016, China,  xwcao@nuaa.edu.cn.  X. Cao is supported by the National Natural Science Foundation of China (No. 11771007).}
\date{}

\maketitle
\thispagestyle{empty}

\abstract{ As a subclass of linear codes, cyclic codes have efficient encoding and decoding algorithms, so they are widely used in many areas such as  consumer electronics, data storage systems and communication systems.    In this paper, we give a general construction of optimal $p$-ary cyclic codes which leads to three explicit constructions. In addition, another  class of  $p$-ary optimal cyclic codes  are presented.}



\section{Introduction}
Let $p$ be a  prime. Denote by  $\mathbb{F}_{p}$ the finite field with $p$ elements.  An $[n,l,d]$ linear code $\mathcal{C}$ over $\mathbb{F}_{p}$ is a linear subspace of $\mathbb{F}_{p}^{n}$ with dimension $l$ and minimum Hamming distance $d$.
Moreover,  $\mathcal{C}$ is called cyclic if   any $(c_{0}, c_{1},  \ldots, c_{n-1}) \in \mathcal{C}$ implies $(c_{n-1}, c_{0},  \ldots, c_{n-2}) \in \mathcal{C}$.   
It is well known that each codeword $(c_{0}, c_{1},  \ldots, c_{n-1}) \in \mathbb{F}_{p}^{n}$  can be regarded as  a polynomial $c_{0}+ c_{1}x+\cdots+c_{n-1}x^{n-1} \in \mathbb{F}_{p}[x]/(x^{n}-1)$.
Then a linear code $\mathcal{C}$ in $\mathbb{F}_{p}^{n}$ is cyclic if and only if $\mathcal{C}$ is an ideal of the polynomial residue class ring $\mathbb{F}_{p}[x]/(x^{n}-1)$.  Since each ideal of the ring $\mathbb{F}_{p}[x]/(x^{n}-1)$ is  principal, every cyclic code corresponds to a principal ideal $(g(x))$ of  the multiples of a polynomial $g(x)$ which is the monic polynomial of lowest degree in the ideal.   $g(x)$ is called the generator polynomial,   $h(x)=(x^{n}-1)/g(x)$ is called the parity-check polynomial of the code $\mathcal{C}$. If $g(x)$ can be reduced to a product of $r$ different irreducible polynomials in $\mathbb{F}_{p}[x]$, then $\mathcal{C}$ is described to have $r$ zeros.

Cyclic codes  have   wide  practical applications in many areas   as they have efficient encoding and decoding algorithms.
Moreover, they also have wide applications in cryptography and sequence design. So in the past few decades,  much progress has been made on cyclic codes.
It is worth mentioning that in recent years, many scholars are interested in studying optimal cyclic codes over finite fields according to  Sphere packing bound \cite{11}.  For example, C. Carlet et al. \cite{4}  constructed optimal ternary cyclic codes with minimum distance 4  by using perfect nonlinear monomials. In 2013, C. Ding and T. Helleseth \cite{6}
obtained some optimal ternary cyclic codes by employing almost perfect nonlinear monomials
and some other monomials over $\mathbb{F}_{p^{m}}$. They also presented nine open problems. After that,
two of the open problems were solved by N. Li et al.  \cite{32,23}. In \cite{32}, the authors also presented several classes of optimal cyclic codes with parametes $[3^{m}-1, 3^{m}-1-2m, 4]$ and $[3^{m}-1, 3^{m}-2-2m, 5]$.  In 2014, Z. Zhou and
C. Ding \cite{19} gave a class of optimal ternary cyclic codes.  C. Fan et al.\cite{8}  obtained a new class of optimal ternary cyclic codes with minimum
distance four.  Furthermore, they also discussed the weight of   duals of them.  What's more, L. Wang and G. Wu \cite{35} listed four classes of optimal ternary cyclic codes. Afterwards,  H. Yan et al. \cite{3}   also obtained a new class of optimal ternary cyclic codes and discussed the weight of their duals. Different from these work, G. Xu et al. \cite{41} constructed optimal $p$-ary  cyclic codes by making use of monomilas. More generally, C. Ding and S. Ling \cite{5} proposed a $q$-polynomial method for the construction of cyclic codes.   Recently,   W. Fang et al. \cite{9} used $q$-polynomials to construct a class of $[2(q^{m}-1)/(q-1), 2(q^{m}-1)/(q-1), 4]$ constacyclic codes which are optimal. In addition, Y. Zhou et al. \cite{46} constructed several classes of optimal negacyclic codes $[(5^{m}-1)/2, (5^{m}-1)/2, 4]$ over $\mathbb{F}_{5}$. For information on the related topics, the reader is referred to \cite{51,1,15,7,10,31,22,13,43,14,45} and the references therein.


The rest of this paper is  organized as follows. Some preliminaries  will be introduced in Section 2.  Optimal $p$-ary cyclic codes are discussed in Section 3. Section 4 concludes this paper.

\section{Preliminaries}
In this section, we will fix some basic notation for this paper and introduce  $p$-cyclotomic cosets that will be used in   subsequent sections.
Throughout this paper, we will use the following notation unless otherwise stated.
\begin{itemize}
  \item $ p$ is an odd prime.
  \item  $n=2(p^{m}-1)/(p-1)$, where $m$ is a positive integer.
   \item $\pi$ is a primitive element of  $\mathbb{F}_{p^{m}}$.
    \item $\sigma \in \mathbb{F}_{p^{m}}$ is  a primitive $n$-th root of unity.
    \item   $m_{i}(x)$ is the minimal polynomial of $\sigma^{i}$ over $\mathbb{F}_{p}$.
  \item  $\mathbb{Z}_{n}=\{ 0,1,2,\ldots,n-1\}$ associated with the integer addition modulo $n$ and integer multiplication modulo $n$ operations.
  \item  $\Pi=\{\sigma, \sigma^{2}, \cdots,  \sigma^{n-1}\}$.
\end{itemize}
For any integer $j$, $0 \leq j \leq n-1$, the $p$-cyclotomic coset modulo $n$ containing $j$ is defined by
\[C_{j}=\{j, pj, p^{2}j,\cdots,p^{l_{j}-1}j\}\subset \mathbb{Z}_{n}\]
where $l_{j}$ is the minimal positive integer such that $p^{l_{j}}j\equiv j  \pmod n$, and is called the \emph{length} of $C_{j}$  which is denoted by $|C_{j}|$. The smallest integer in  $C_{j}$ is called the  \emph{coset leader} of $C_{j}$. Let $\Gamma$ be the set of all coset leaders. By definition, we have $\bigcup_{j \in \Gamma} C_{j}= \mathbb{Z}_{n}$.
In the following, we give a lemma about the length of cyclotomic cosets.
\begin{lemma}[\cite{52}]\label{Le:2.4}
For any   integer $j$, $0\leq j \leq n-1$ with $\gcd(j,n)=d$, the length of $C_{j}$ is equal to $m$ if $1\leq d \leq 2(p+1)$.
\end{lemma}
Let $\mathcal{C}_{p}(u,v)$ be the cyclic code of length $n$ over  $\mathbb{F}_{p}$ with generator polynomial $m_{u}(x)m_{v}(x)$ where $u$, $v$ are in $\mathbb{Z}_{n}$ and they are not in the same cyclotomic coset. The following two lemmas  are important to prove the main results of this paper.
\begin{lemma} [\cite{52}]\label{Le:2.1}
The minimum distance of $\mathcal{C}_{p}(u,v)$ is no less than $3$ if $gcd(v-u, n)=1$.
\end{lemma}
\begin{lemma} [\cite{52}]\label{Le:2.2}
Let $v \notin C_{1}$ and $|C_{v}|=m$. Then $\mathcal{C}_{p}(1,v)$ is optimal with parameters $[n, n-2m, 4]$ if the following conditions are satisfied:
\begin{enumerate}[$1)$]
  \item  $\gcd(v-1, n)=1$;
  \item $v\equiv 1 \pmod{\frac{p-1}{2}}$; and
  \item the equations $(x+\alpha)^{v}\pm(x^{v}+\alpha)=0$ have no solutions in $\Pi$ for any  $\alpha$ in $\mathbb{F}_{p}^{*}$.
  \end{enumerate}
\end{lemma}
\remark The conditions $1)$ and $2)$ implies that $v$ is even and $p \equiv 3 \pmod{4}$ since $n=2(p^{m}-1)/(p-1)$ is even.

\noindent In fact, the lemma above can be reduced to   the following lemma.
\begin{lemma}\label{Le:2.3}
Let $v \notin C_{1}$ and $|C_{v}|=m$. Then  $\mathcal{C}_{p}(1,v)$ is optimal with parameters $[n, n-2m, 4]$ if the following conditions are satisfied:
\begin{enumerate}[$1)$]
  \item  $\gcd(v-1, n)=1$;
  \item  $v\equiv 1 \pmod{\frac{p-1}{2}}$; and
  \item the equations $(x+\alpha)^{v}\pm(x^{v}+\alpha)=0$ have no solutions in $\Pi\setminus\{\sigma^{\frac{n}{2}}\}$, i.e., $\Pi\setminus\{-1\}$ for any  $\alpha$ in $\mathbb{F}_{p}^{*}$.
  \end{enumerate}
\end{lemma}
\begin{proof}
By Lemma \ref{Le:2.2}, we only need to prove that  $-1$ is not a solution of $(x+\alpha)^{v}\pm(x^{v}+\alpha)=0$ for any $\alpha$ in $\mathbb{F}_{p}^{*}$. Otherwise, there is an element $\alpha_{0} \in \mathbb{F}_{p}^{*}\backslash\{1\}$ such that \begin{equation}\label{Eq:2.1}(-1+\alpha_{0})^{v}\pm((-1)^{v}+\alpha_{0})=0.
\end{equation}
By the conditions  $1)$ and $2)$, suppose  $v=\frac{p-1}{2}k+1$, then $k$ is an odd integer. Note that $\alpha_{0}-1  \in \mathbb{F}_{p}^{*}$, then (\ref{Eq:2.1}) becomes $(-1+\alpha_{0})\mp(1+\alpha_{0})=0$ which is impossible.
\end{proof}
\section{New optimal $p$-ary codes with parameters $[n, n-2m, 4]$}
In this section, we will give two new classes of $p$-ary cyclic codes $\mathcal{C}_{p}(1,v)$ with parameters $[n, n-2m, 4]$ which are optimal according to  Sphere packing bound \cite{11}.

\noindent\textbf{3.1 The first class of optimal $p$-ary codes }

First, we consider the codes $\mathcal{C}_{p}(1,v)$ with $v=p^{k}+1$, where $k$ is an integer such that $0\leq k \leq m-1$.
\begin{theorem}
Let $m > 2$ be an odd integer. Let $v=p^{k}+1$ where $k$ is an integer such that $0\leq k \leq m-1$. Then $\mathcal{C}_{p}(1,v)$ is optimal with parameters $[n, n-2m, 4]$ if $\gcd(m, k)=\gcd(m, p-1)=1$ and $\frac{p-1}{2}\mid k$.
\end{theorem}
\begin{proof}
It is easy to check that $\gcd(v-1, n)=1$ and $v \notin C_{1}$. Note that $\frac{p-1}{2}\mid k$, $v\equiv k+1\equiv 1 \pmod {\frac{p-1}{2}}$.
 In addition, $\gcd(v, \frac{n}{2})=\gcd(p^{k}+1, \frac{p^{m}-1}{p-1})=1$ since $\gcd(p^{k}+1, p^{m}-1)=2$ when $m$ is odd. Hence, $\gcd(v, n)=2$. Then $|C_{v}|=m$ by Lemma \ref{Le:2.4}.

In the following, we will prove the equations $(x+\alpha)^{v}\pm(x^{v}+\alpha)=0$ have no solutions in $\Pi\setminus\{-1\}$ for any  $\alpha$ in $\mathbb{F}_{p}^{*}$ when $\gcd(m, p-1)=1$.
Suppose that there is a solution $x_{0}$ in $\Pi\setminus\{-1\}$. Then
\begin{equation}\label{Eq:3.8}
(x_{0}+\alpha)^{v}=\pm(x_{0}^{v}+\alpha).
\end{equation}
Taking $(p^{k}-1)$-th power on both sides of (\ref{Eq:3.8}), we have
\[(x_{0}+\alpha)^{p^{2k}-1}=(x_{0}^{p^{k}+1}+\alpha)^{p^{k}-1}
\]
which can be reduced to
\[x_{0}(x_{0}^{p^{2k}-1}-1)(x_{0}^{p^{k}}-1)=0.
\]
Hence, $x_{0}^{p^{2k}-1}=1$ or $x_{0}^{p^{k}}=1$ which is impossible, since $\gcd(p^{k}, n)=1$, $\gcd(p^{2k}-1,n)=2\gcd(p-1,m)=2$ when $\gcd(m, k)=1$ and $\gcd(m, p-1)=1$. By Lemma \ref{Le:2.3}, the result follows.
\end{proof}
\begin{example}
Let $p=5$, $m=3$, $k=2$ .  Then $v=10$ and the code $\mathcal{C} _{1,10}$ is an optimal   cyclic code with parameters $[62,50,4]$ and generator polynomial  \[  x^6 + 3x^5 + 2x^3 + 2x^2 + 4.\]
\end{example}
\noindent\textbf{3.2 The second   class of optimal $p$-ary codes}

In the following, we consider the codes $\mathcal{C}_{p}(1,v)$ with
\begin{equation}\label{Eq:3.1}
(p^{t}-1)v\equiv p^{s}-p^{h}\pmod{p^{m}-1}
\end{equation}
where $t, s, h, m$ are integers such that $0 \leq t, s, h \leq m-1$.
\begin{lemma}\label{Le:3.1}
Let $t, s, h, m$ be integers such that $0 \leq t, s, h \leq m-1$. Then the congruence equation (\ref{Eq:3.1}) has solutions for $v$ such that $C_{v}=m$ if $\gcd(m, t)=\gcd(m, s-h)=1$.
\end{lemma}
\begin{proof}
Let $\tau =\gcd(p^{m}-1, p^{t}-1)$ and $r =\gcd(p^{m}-1, p^{s}-p^{h})$. Then $\tau =r =p-1$  since $\gcd(m, t)=\gcd(m, s-h)=1$.  So $\tau \mid (p^{s}-p^{h})$  which implies (\ref{Eq:3.1}) has solutions. Since $r =\gcd(p^{m}-1, (p^{t}-1)v)=p-1$, it follows that $\gcd(v, n)=1$ or $2$. Then by Lemma \ref{Le:2.4}, $C_{v}=m$.
\end{proof}
By Lemma \ref{Le:2.3} and \ref{Le:3.1}, we have the following result.
\begin{theorem}\label{Th:3.1}
Let $t, s, h$ be integers such that $0 \leq t, s, h \leq m-1$ and $\gcd(m, t)=\gcd(m, s-h)=1$ where $m> 2$ is an integer with $\gcd(m, p-1)\mid 2$. Let $v$ be a solution of (\ref{Eq:3.1}). Then $\mathcal{C}_{p}(1,v)$ is optimal with parameters $[n, n-2m, 4]$ if $\gcd(p^{h}-v, n)=\gcd(v-1, n)=1$ and $v\equiv 1 \pmod{\frac{p-1}{2}}$.
\end{theorem}
\begin{proof}
First, we prove $v \notin C_{1}$. Suppose that $v \in C_{1}$, then $p^{\lambda} \equiv v \pmod{n}$ for some integer $\lambda$. Hence $n \mid (p^{\lambda}- v)$. Then $v$ is odd since $p$ is odd and $n$ is even which contradicts the condition $\gcd(v-1, n)=1$.

Next, we prove the equations $(x+\alpha)^{v}\pm(x^{v}+\alpha)=0$ have no solutions in $\Pi\setminus\{\sigma^{\frac{n}{2}}\}$  for any  $\alpha$ in $\mathbb{F}_{p}^{*}$. Suppose that there is a solution $x_{0} \in \Pi\setminus\{-1\}$. Then
\begin{equation}\label{Eq:3.9}(x_{0}+\alpha)^{v}=\pm(x_{0}^{v}+\alpha).
\end{equation}
Taking $p^{t}-1$-th power on both sides of (\ref{Eq:3.9}), we have
\[(x_{0}+\alpha)^{v(p^{t}-1)}=(x_{0}^{v}+\alpha)^{p^{t}-1}.
\]
 By (\ref{Eq:3.1}), the equation above  can be reduced to
 \[x_{0}^{p^{s}}+x_{0}^{v}=x_{0}^{vp^{t}}+x_{0}^{p^{h}}
\]
by easy calculation.
Note that $p^{t}v \equiv v+p^{s}-p^{h} \pmod{p^{m}-1}$, then the equation above becomes
\[x_{0}^{v}(x_{0}^{p^{s}-p^{h}}-1)(x_{0}^{p^{h}-v}-1)=0.
\]
Hence, $x_{0}^{p^{s}-p^{h}}=1$ or $x_{0}^{p^{h}-v}=1$ which is impossible, since $\gcd(p^{h}-v, n)=1$, $\gcd(p^{s}-p^{h},n)=2$ when $\gcd(m, s-h)=1$ and $\gcd(m, p-1)\mid 2$.
\end{proof}
By Theorem \ref{Th:3.1}, we can have three concrete constructions as follows.
\begin{corollary}\label{Cor:3.1}
Let $m> 2$ be an odd integer such that $\gcd(m, p-1)=1$. Let $t=1, s=0, h=1$ and $v=\frac{n}{2}-1$ be a solution of (\ref{Eq:3.1}). Then $\mathcal{C}_{p}(1,v)$ is optimal with parameters $[n, n-2m, 4]$ if $\frac{p-1}{2} \mid (m-2)$.
\end{corollary}
\begin{proof}
It is clear that $\gcd(m, t)=\gcd(m, s-h)=1$.
$\gcd(p^{h}-v, n)=\gcd(p+1-\frac{n}{2}, n)=\gcd(p+1, \frac{n}{2})=1$ since $m$ is odd. Note that $\gcd(v-1, n)=\gcd(\frac{n}{2}-2, n)=1$.
Furthermore, $v-1 \equiv m-2 \pmod{\frac{p-1}{2}}$. Then $v \equiv 1\pmod{\frac{p-1}{2}}$ since $\frac{p-1}{2} \mid (m-2)$. By Theorem \ref{Th:3.1}, the result follows.
\end{proof}
\remark If $p=3$ in Corollary \ref{Cor:3.1}, then the result reduced to $\mathcal{C}_{3}(1,\frac{3^{m}-3}{2})$ is optimal if $m \geq 3$ is odd which generalizes a result in \cite{6}.
\begin{example}
Let $p=7$, $m=5$.  Then $v=2800$ and the code $\mathcal{C} _{1,2800}$ is an optimal  cyclic code with parameters $[5602,5592,4]$ and generator polynomial  \[ x^{10} + 6x^9 + 2x^8 + 2x^7 + 6x^6 + 3x^5 + x^4 + 2x^3 + 5x^2 + 6x + 6.\]
\end{example}
\begin{corollary}\label{Cor:3.2}
Let $m> 2$ be  an odd integer such that $\gcd(m, p-1)=1$. Let $t=1, s\geq 2, h=0$ and $v=\frac{p^{s}-1}{p-1}$ be a solution of (\ref{Eq:3.1}). Then $\mathcal{C}_{p}(1,v)$ is optimal with parameters $[n, n-2m, 4]$ if $s$ is even, $\gcd(m, s)=\gcd(m, s-1)=1$ and $\frac{p-1}{2} \mid (s-1)$.
\end{corollary}
\begin{proof}
It is clear that $\gcd(m, t)=\gcd(m, s-h)=1$  since $\gcd(m, s)=1$.
$\gcd(p^{h}-v, n)=\gcd(v-1, n)=\gcd(\frac{p^{s}-p}{p-1}, n)=\gcd(\frac{p^{s-1}-1}{p-1}, n)=1$ due to $s$ is even and $\gcd(m, s-1)=1$.
Furthermore, $v-1 \equiv s-1 \pmod{\frac{p-1}{2}}$. Note that $\frac{p-1}{2} \mid (s-1)$, then $v \equiv 1\pmod{\frac{p-1}{2}}$. By Theorem \ref{Th:3.1}, the result follows.
\end{proof}
\begin{example}
Let $p=7$, $m=5$, $s=4$.  Then $v=400$ and the code $\mathcal{C} _{1,400}$ is  optimal  with parameters $[5602,5592,4]$ and generator polynomial  \[  x^{10} + 6x^9 + 2x^8 + 2x^7 + 6x^6 + 3x^5 + x^4 + 2x^3 + 5x^2 + 6x + 6.\]
\end{example}
\begin{corollary}\label{Cor:3.3}
Let $m> 2$ be an odd integer such that $\gcd(m, p-1)=1$. Let $t=1, s\geq 2, h=0$ and $v=\frac{n}{2}+\frac{p^{s}-1}{p-1}$ be a solution of (\ref{Eq:3.1}). Then $\mathcal{C}_{p}(1,v)$ is optimal with parameters $[n, n-2m, 4]$ if $s$ is odd, $\gcd(m, s)=\gcd(m, s-1)=1$ and $\frac{p-1}{2} \mid (m+s-1)$.
\end{corollary}
\begin{proof}
It is clear that $\gcd(m, t)=\gcd(m, s-h)=1$  since $\gcd(m, s)=1$.
 Since $m$ and $s$ are both odd, $\gcd(p^{h}-v, n)=\gcd(v-1, n)=\gcd(\frac{n}{2}+\frac{p^{s}-1}{p-1}-1, n)=\gcd(\frac{n}{2}+\frac{p^{s}-1}{p-1}-1, \frac{n}{2})=\gcd(\frac{p^{s}-1}{p-1}-1, \frac{n}{2})$. Note that $\gcd(m, s-1)=1$, then $\gcd(\frac{p^{s}-1}{p-1}-1, \frac{n}{2})=\gcd(\frac{p^{s-1}-1}{p-1}-1, \frac{p^{m}-1}{p-1})=1$.
Furthermore, $v-1 \equiv m+s-1 \pmod{\frac{p-1}{2}}$. Then $v \equiv 1\pmod{\frac{p-1}{2}}$ due to $\frac{p-1}{2} \mid m+s-1$. By Theorem \ref{Th:3.1}, the result follows.
\end{proof}
\begin{example}
Let $p=3$, $m=7$, $s=3$.  Then $e=86$ and the code $\mathcal{C} _{1,86}$ is   optimal with parameters $[728,714,4]$ and generator polynomial  \[  x^{14} + x^{13} + 2x^{10} + x^9 + 2x^8 + x^6 + x^5 + x^4 + 2x^3 + 2.\]
\end{example}
\section{Conclusions}
In this paper, we give a general construction of optimal $p$-ary cyclic codes which leads to three explicit constructions. In addition, another class of  optimal cyclic codes $\mathcal{C}_{p}(1,v)$ with $v=p^{k}+1$ are presented.

\end{document}